\documentclass{amsart}

\usepackage{dsfont,amsthm,amsmath}

\newcommand{\R}{\mathds{R}}
\newcommand{\Z}{\mathds{Z}}
\newcommand{\N}{\mathds{N}}
\newcommand{\1}{\mathds{1}}
\newcommand{\from}{\colon}
\renewcommand{\P}{\mathds{P}}

\newcommand{\loc}{\mathrm{loc}}
\newcommand{\lu}{\mathrm{unif}}
\newcommand{\unif}{\mathrm{unif}}
\newcommand{\Mloc}{\mathcal{M}_{\loc}(\R)}
\newcommand{\M}{\mathcal{M}_{\loc,\unif}(\R)}

\providecommand{\abs}[1]{\left\lvert#1\right\rvert}
\providecommand{\norm}[1]{\left\lVert#1\right\rVert}
\providecommand{\set}[1]{\left\{ #1\right\}}

\makeatletter

\theoremstyle{plain} 
\newtheorem{theorem}{Theorem}[section]

\newtheorem{lemma}[theorem]{Lemma}
\newtheorem{proposition}[theorem]{Proposition}
\theoremstyle{definition}
\newtheorem{example}[theorem]{Example}
\newtheorem{definition}{Definition}

\makeatother

\begin{document}

\title{Almost sure purely singular continuous spectrum for quasicrystal models}

\author{C. Seifert$^*$}

\address{Technische Universit\"at Hamburg-Harburg, Institut f\"ur Mathematik,\\
21073, Hamburg, Germany\\
$^*$E-mail: christian.seifert@tuhh.de\\
http://www.mat.tuhh.de/home/cseifert}

\begin{abstract}
We review recent developments in the spectral theory of continuum one-dimensional quasicystals, yielding purely singular continuous spectrum for these Schr\"odinger operators.
Allowing measures as potentials we can generalize some results to very singular
 potentials, including Kronig-Penney type models.
\end{abstract}



\maketitle
\section{Introduction}
\label{sec:introduction}

The aim of this short note is to review continuum models for quasicrystals in one dimension.
Such models have been studied intensively during the last decades.

Here, we allow for measures $\mu$ as potentials, i.e.\ we are studying the spectral theory of self-adjoint Schr\"odinger operators of the form
\[H_\mu:= -\Delta+\mu\]
in $L_2(\R)$. In order to deal with quasicrystal models the potential $\mu$ 
should satisfy certain properties. Quasicrystals can be considered as media 
between order (i.e., perfect crystals) and disorder (i.e., amorphous media).
On the one hand, quasicrystals are not too far from periodic crystals, so $\mu$
 should be locally close to being periodic, but still globally aperiodic. 
Furthermore, $\mu$ should satisfy a finite local complexity condition (in a 
suitable sense for measures).
On the other hand, quasicrystals share features from disordered media, so the 
corresponding Hamiltonians may be considered as random operators with some 
ergodicity assumptions.
Thus, one is tempted to expect neither pure-point spectrum nor absolutely continuous spectrum for such models.

Combining results on the absolutely continuous spectrum obtained in Ref.~\cite{KlassertLenzStollmann2011} 
(which make use of Remlings Oracle Theorem in Ref.~\cite{Remling2007})
with Gordon-type arguments for singular potentials as in Ref.~\cite{SeifertVogt2014} 
(see also Ref.~\cite{Gordon1976, DamanikStolz2000, Seifert2011} for earlier results) we will conclude generic (in the sense of an ergodic measure)
purely singular continuous spectrum for such operators. We also give a scheme how to produce examples by means of discrete subshifts and a suspension type construction.

\section{Schr\"odinger operators with measures as potentials}
\label{sec:Schroedinger}

We say that
\[
  \mu \from \set{B\subseteq\R;\;B \text{ is a bounded Borel set}} \to \R
\]
is a \emph{local measure} if $\1_K\mu := \mu(\cdot\cap K)$ is a
signed Radon measure for any compact set $K\subseteq \R$.
Then there exist a (unique) nonnegative Radon measure $\nu$ on $\R$ and
a measurable function $\sigma\from\R\to \R$ such that $\abs{\sigma} = 1$ $\nu$-a.e.\ and
$\1_K\mu = \1_K\sigma\nu$ for all compact sets $K\subseteq \R$. The \emph{total variation} of $\mu$ is defined by $\abs{\mu}:=\nu$.
Let $\Mloc$ be the space of all local measures on $\R$.

A local measure $\mu\in \Mloc$ is called \emph{uniformly locally bounded} if
\[
  \norm{\mu}_\lu := \sup_{a\in\R} \abs{\mu}((a,a+1]) < \infty.
\]
Let $\M$ denote the space of all uniformly locally bounded local measures.
The space $\M$ naturally extends $L_{1,\loc,\unif}(\R)$ to measures.

For $\mu\in\M$ let $H_\mu$ be the maximal operator in $L_2(\R)$ associated with $-\Delta+\mu$ (in the distributional sense): define
\begin{align*}
  D(H_\mu) & := \set{u\in L_2(\R)\cap C(\R);\; -u'' + u\mu\in L_2(\R)},\\
  H_\mu u & := -u'' + u\mu.
\end{align*}
Then $H_\mu$ is self-adjoint.
Realizations of $-\Delta+\mu$ can also be defined via form methods (see Ref.~\cite{KlassertLenzStollmann2011}) or Sturm-Liouville theory (see Ref.~\cite{BenAmorRemling2005}).
As shown in Theorem 3.6 of Ref.~\cite{SeifertVogt2014} all these methods lead to the same self-adjoint operator.

\section{Spectral properties of quasicystalline models}

Let us recall some definitions from Ref.~\cite{KlassertLenzStollmann2011} concerning notions of finite local complexity for measures.

\begin{definition}
  A \emph{piece} is a pair $(\nu,I)$ consisting of an interval $I\subseteq\R$ with positive length 
  $\lambda(I) >0$ (which is then called the \emph{length} of the piece) and a signed (local) measure $\nu$ on $\R$ supported on $I$. 
  We abbreviate pieces by $\nu^I$. 
  A \emph{finite piece} is a piece of finite length. 
  We say $\nu^I$ \emph{occurs} in a signed (local) measure $\mu$ at $x\in\R$, if $\1_{[x,x+\lambda(I)]}\mu$ 
  is a translate of $\nu$.

	The \emph{concatenation} $\nu^I=\nu_1^{I_1}\mid \nu_2^{I_2}\mid \ldots$ of a finite or countable family 
    $(\nu_j^{I_j})_{j\in N}$, with $N=\set{1,2,\ldots,\abs{N}}$ (for $N$ 
finite) or $N=\N$ (for $N$ infinite), of finite pieces is defined by
	\begin{align*}
		I & = \left[\min I_1,\min I_1 + \sum_{j\in N} \lambda(I_j)\right],\\
		\nu & = \nu_1+\sum_{j\in N,\,j\geq 2} \nu_j\Big(\cdot-\Big(\min I_1 + \sum_{k=1}^{j-1} \lambda(I_k) - \min I_j\Big)\Big).
	\end{align*}
	We also say that $\nu^I$ is \emph{decomposed} by $(\nu_j^{I_j})_{j\in N}$.
\end{definition}

\begin{definition}
	Let $\mu$ be a signed (local) measure on $\R$. We say that $\mu$ has the \emph{finite decomposition property} (f.d.p.), 
	if there exist a finite set $\mathcal{P}$ of finite pieces (called the \emph{local pieces}) 
	and $x_0\in\R$, such that $\1_{[x_0,\infty)}\mu^{[x_0,\infty)}$ is a translate of a 
	concatenation $v_1^{I_1}\mid \nu_2^{I_2}\mid\ldots$ with $\nu_j^{I_j}\in\mathcal{P}$ for all $j\in\N$. 
	Without restriction, we may assume that $\min I = 0$ for all $\nu^I\in \mathcal{P}$.

	A signed (local) measure $\mu$ has the \emph{simple finite decomposition property} (s.f.d.p.), if it has the f.d.p.~with a decomposition such that there is $\ell>0$ with the following property: Assume that the two pieces
	\[\nu_{-m}^{I_{-m}} \mid \ldots \mid \nu_{0}^{I_{0}} \mid \nu_{1}^{I_{1}} \mid \ldots \mid \nu_{m_1}^{I_{m_1}} \quad \text{and} \quad 
	 \nu_{-m}^{I_{-m}} \mid \ldots \mid \nu_{0}^{I_{0}} \mid \mu_{1}^{J_{1}} \mid \ldots \mid \mu_{m_2}^{J_{m_2}}\]
	occur in the decomposition of $\mu$ with a common first part $\nu_{-m}^{I_{-m}} \mid \ldots \mid \nu_{0}^{I_{0}}$ of length at least $\ell$ and such that
	\[\1_{[0,\ell)}(\nu_{1}^{I_{1}} \mid \ldots \mid \nu_{m_1}^{I_{m_1}}) = \1_{[0,\ell)}(\mu_{1}^{J_{1}} \mid \ldots \mid \mu_{m_2}^{J_{m_2}}),\]
	where $\nu_j^{I_j}$, $\mu_k^{J_k}$ are pieces from the decomposition (in particular, all belong to $\mathcal{P}$ and start at $0$) and the latter two concatenations are of lengths at least $\ell$. Then 
	\[\nu_1^{I_1} = \mu_1^{J_1}.\]
\end{definition}

We can formulate a striking spectral consequence for Schr\"odinger operators with potentials having the s.f.d.p.

\begin{theorem}[{see \cite[Theorem 7.1]{LenzSeifertStollmann2014}}]
  Let $\mu\in\M$ such that $\mu$ and $\mu(-(\cdot))$ have the s.f.d.p.\ and assume that neither $\mu$ nor $\mu(-(\cdot))$ are eventually periodic. Then $\sigma_{ac}(H_\mu) = \emptyset$.
\end{theorem}

As noted in the introduction Schr\"odinger operators modelling quasicrystals 
can be thought of as random operators. To this end, let $\Omega\subseteq \M$ be
 $\norm{\cdot}_\lu$-bounded and closed with respect to the vague topology.
Then $\Omega$ is compact and hence metrizable (see \cite[Proposition 4.1.2]{Seifert2012}). Furthermore, assume that $\Omega$ is translation invariant, i.e., we have a group action 
$\alpha\from\R\times\Omega\to\Omega$, $\alpha(t,\omega):=\alpha_t(\omega):= \omega(\cdot+t)$ acting on $\Omega$. Note that $\alpha$ is continuous.
For $\omega\in\Omega$ define $H_\omega$ as above. Note that 
$(H_\omega)_{\omega\in\Omega}$ is a random operator (see \cite[Theorem 
3.6]{LenzSeifertStollmann2014}) and it is an ergodic operator family if 
$(\Omega,\alpha)$ is ergodic with ergodic measure $\P$.
Recall that $(\Omega,\alpha)$ is called minimal if every orbit $\set{\alpha_t(\omega);\; t\in\R}$ is dense in $\Omega$. Both ergodicity and minimality are suitable assumptions for quasicrystalline models.

\begin{proposition}[{see \cite[Theorem 5.1]{KlassertLenzStollmann2011}}]
\label{prop:Sigma_ac_empty}
  Let $(\Omega,\alpha)$ be minimal, having the
    s.f.d.p.~(i.e.\ for every $\omega\in\Omega$: $\omega$ and $\omega(-(\cdot))$ has s.f.d.p.),
  and aperiodic (i.e.\ there exists $\omega\in\Omega$ which is not periodic). 
  Then $\sigma_{ac}(H_\omega) = \emptyset$ for all $\omega\in\Omega$.
\end{proposition}

\begin{proof}
  Assume that $\set{\omega\in\Omega;\; \sigma_{ac}(H_\omega)\neq\emptyset}$ in nonempty. 
  Then by \cite[Theorem 4.1]{KlassertLenzStollmann2011}, the set $\set{\omega\in\Omega;\; \text{$\omega$ or $\omega(-(\cdot))$ is 
  eventually periodic}}$ 
  is nonempty as well. W.l.o.g.~assume that $\omega$ is periodic for $t\geq t_0$ with period $p$. 
  By closedness of $\Omega$, 
  \[\tilde{\omega}:= \lim_{t\to\infty} \alpha_{t}(\omega) = \lim_{t\to\infty} \omega(\cdot+t)\in \Omega\]
  and $\tilde{\omega}$ is periodic with period $p$.
  By minimality, for any $\omega'\in\Omega$ there exists $(t_n)$ in $\R$ such that
  $\alpha_{t_n}(\tilde{\omega})\to \omega'$. Since $\tilde{\omega}$ is periodic and $\alpha$ is continuous, we arrive at
  \[\alpha_p(\omega') = \alpha_p\left(\lim_{n\to\infty} \alpha_{t_n}(\tilde{\omega})\right) = \lim_{n\to\infty} \alpha_{t_n}\alpha_p(\tilde{\omega}) = \omega'.\]
  So, every $\omega\in\Omega$ must be periodic with the same period, a contradiction.
\end{proof}

Let us now focus on the pure point spectrum. As mentioned in the introduction, potentials modelling quasicrystals are close to periodic potentials.
The following condition due to Kaminaga\cite{Kaminaga1996} suits very 
well for ``local periodicity''.

\begin{definition}
  Let $(\Omega,\alpha,\P)$ be ergodic. We say that it satisfies condition (K) 
  if there exists $(p_n)$ in $(0,\infty)$ with $p_n\to \infty$ such that
  \[G_n:= \set{\omega\in\Omega;\; \1_{[0,p_n]}\omega = \1_{[0,p_n]}\alpha_{p_n}(\omega) = \1_{[0,p_n]}\alpha_{-p_n}(\omega)} \quad(n\in\N)\]
  satisfies
  $\limsup_{n\to\infty} \P(G_n) > 0$.
\end{definition}
Note that the set $G_n$ contains all measures $\omega\in\Omega$ for which the three pieces of the intervals $[-p_n,0]$, $[0,p_n]$ and $[p_n,2p_n]$ are equal, i.e.\ which are locally close to being periodic.

\begin{lemma}
\label{lem:(K)_absence_ev}
Let $(\Omega,\alpha,\P)$ be ergodic satisfying (K). Then for $\P$-a.a.~$\omega\in\Omega$ we have $\sigma_{pp}(H_\omega) = \emptyset$,
i.e., $H_\omega$ does not have any eigenvalues $\P$-a.s.
\end{lemma}

\begin{proof}
  Let 
  $\Omega_c:=\set{\omega\in\Omega;\; \sigma_{pp}(H_\omega) = \emptyset}$.
  Note that $\Omega_c$ is $\alpha$-invariant.
  Ergodicity implies $\P(\Omega_c)\in \set{0,1}$. 
  By a Gordon type argument, see \cite[Corollary 5.5]{SeifertVogt2014}, we have
  \[G:= \limsup_{n\to\infty} G_n = \bigcap_{n\in\N} \bigcup_{k=n}^\infty G_k\subseteq \Omega_c.\]
  Hence,
  \[\P(\Omega_c) \geq \P(G) = \P(\limsup_{n\to\infty} G_n)\geq \limsup_{n\to\infty} \P(G_n)>0.\]
  We conclude that $\P(\Omega_c) = 1$.
\end{proof}

We can now state our main theorem: the almost sure purely singular continuous spectrum for quasicrystalline models.

\begin{theorem}
  Let $(\Omega,\alpha,\P)$ be ergodic and minimal, having the s.f.d.p., aperiodic, and satisfying (K). Then $H_\omega$ has purely singular continuous spectrum for $\P$-almost all $\omega\in\Omega$.
\end{theorem}

\begin{proof}
  By Proposition \ref{prop:Sigma_ac_empty} we obtain that $\sigma_{ac}(H_\omega) = \emptyset$ for all $\omega\in\Omega$.
  Lemma \ref{lem:(K)_absence_ev} yields $\sigma_{pp}(H_\omega) = \emptyset$ for $\P$-a.a.\ $\omega\in\Omega$.
  Hence, for $\P$-a.a.~$\omega\in\Omega$, $H_\omega$ has purely singular continuous spectrum.
\end{proof}

This result is as far as we know the most general version;
previously known results as stated in \cite[Theorem 
7.4]{KlassertLenzStollmann2011} only work for $L_{1,\loc}$-potentials.

\begin{example}
  Let $A$ be a finite set equipped with the discrete topology. Let $X\subseteq A^\Z$ be closed and invariant under the shift $S\from A^\Z\to A^\Z$, $Sx(n):=x(n+1)$.
  Then $(X,S)$ is called a subshift. 
  We use a suspension type consruction.
  For $a\in A$ choose $\nu_a\in \M$ supported on $[0,l_a]$. For $x\in X$ define $\omega_x\in \M$ by
\[\omega_x:= \sum_{n\in\N_0} \delta_{\sum_{k=0}^{n-1} l_{x(k)}} * \nu_{x(n)} + \sum_{n\in\N} \delta_{\sum_{k=-n}^{-1} -l_{x(k)}} * \nu_{x(-n)}.\]
Let
$
\Omega:= \set{\alpha_{t}(\omega_x);\; x\in X,\, t\in\R}.
$
As shown in \cite[Proposition 8.2]{LenzSeifertStollmann2014} many properties of $(X,S)$ such as ergodicity and minimality may be transferred to $(\Omega,\alpha)$. Furthermore, $(\Omega,\alpha)$ has the s.f.d.p., if at most one of the $\nu_a$'s is a multiple of Lebesgue measure.
Furthermore, condition (K) can also be checked by an analogue condition for 
$(X,S)$, see \cite[Proposition 5]{KlassertLenzStollmann2011}.

Thus, we can generate many examples including Kronig-Penney type models (where the potentials consist of $\delta$-peaks on a quasicrystalline lattice) by choosing appropriate $\nu_a$'s and lengths $l_a$'s.
\end{example}

\bibliographystyle{ws-procs975x65}

\end{document}